%% file: main.tex
\documentclass[11pt,letterpaper]{article}

\input{packages.tex}
\input{environments.tex}

\input{macros.tex}

\begin{document}


\input{trunk/frontpage.tex}

\input{trunk/introduction.tex}

\input{trunk/preliminaries.tex}

\input{trunk/analysis.tex}

\input{trunk/acks.tex}

\printbibliography

\input{trunk/appendix.tex}

\end{document}

%% file: packages.tex
\usepackage[backend=biber,
style=alphabetic,
sorting=nyt,
minalphanames=3,
maxbibnames=99]{biblatex}

\usepackage{amsmath,amsthm,thmtools,amsfonts,amssymb}
\usepackage{dsfont}
\usepackage{mathrsfs}
\usepackage{mathtools}
\usepackage{physics}
\usepackage{braket}
\usepackage{array}
\allowdisplaybreaks

\usepackage[margin=1in]{geometry}
\usepackage{xspace}
\usepackage{xcolor}
\usepackage{subcaption}
\usepackage{enumitem}
\usepackage{microtype}
\usepackage{csquotes}

\usepackage{algorithm}
\usepackage{algpseudocode}

\usepackage{hyperref}
\hypersetup{
    colorlinks=true,
    linkcolor=black,
    citecolor=black,
    filecolor=black,
    urlcolor=black,
}
\usepackage[capitalize]{cleveref}
\crefname{step}{step}{steps}
\Crefname{step}{Step}{Steps}

\usepackage[skins,many]{tcolorbox}

%% file: environments.tex
\newtheorem{theorem}{Theorem}[section]
\newtheorem{lemma}[theorem]{Lemma}
\newtheorem{corollary}[theorem]{Corollary}

\newtheorem{observation}[theorem]{Observation}

\theoremstyle{definition}
\newtheorem{definition}[theorem]{Definition}
\theoremstyle{remark}

%% file: macros.tex
\newcommand{\local}{LOCAL\xspace}

\newcommand{\pr}[1]{\Pr\left[#1\right]}
\newcommand{\expect}[1]{\mathbb{E}\left[#1\right]}
\newcommand{\expectView}[2]{\mathbb{E}_{v,#1,G}\left[#2\right]}

\newcommand{\variance}[1]{\operatorname{Var}\left[#1\right]}
\DeclareMathOperator{\dist}{dist} 
\newcommand{\neighborhood}{\NN}
\DeclareMathOperator{\poly}{poly} 

\newcommand{\view}{\VV}
\DeclareMathOperator{\diam}{diam}

\DeclareMathOperator{\girth}{girth}
\DeclareMathOperator{\indNum}{\alpha}

\newcommand{\nats}{\mathbb{N}}
\newcommand{\natsPos}{\nats_{+}}

\newcommand{\reals}{\mathbb{R}}

\renewcommand{\AA}{\mathcal{A}}

\newcommand{\EE}{\mathcal{E}}

\newcommand{\NN}{\mathcal{N}}

\newcommand{\VV}{\mathcal{V}}

\newcommand{\st}{\ \middle| \ }

\newcommand{\maxDeg}{\Delta}
\newcommand{\algo}{\AA}

\DeclareMathOperator{\bern}{Bern}

\DeclareMathOperator{\rademacher}{Rad}
\newcommand{\M}{\mathsf{M}}
\newcommand{\U}{\mathsf{U}}

\makeatletter
\DeclareMathOperator{\regularlog}{log}
\renewcommand{\log}{\protect\@ifstar{\regularlog^*}{\regularlog}}
\makeatother

\newtcolorbox{myframe}[2][]{%
	breakable,enhanced,colback=white,colframe=black,coltitle=black,
	sharp corners,boxrule=0.4pt,
	fonttitle=\itshape,
	attach boxed title to top left={yshift=-0.3\baselineskip-0.4pt,xshift=2mm},
	boxed title style={tile,size=minimal,left=0.5mm,right=0.5mm,
		colback=white,before upper=\strut},
	title=#2,#1
}


%% file: trunk/frontpage.tex
\begin{flushleft}
    \huge\bf
    New Hardness Results for the LOCAL Model via a Simple Self-Reduction
\end{flushleft}
\smallskip

\newcommand{\myemail}[1]{\,$\cdot$\, {\small #1}}
\newcommand{\myaff}[1]{\,$\cdot$\, {\small #1}\par\medskip}

\newenvironment{myabstract}
{\list{}{\listparindent 1.5em
        \itemindent    \listparindent
        \leftmargin    0cm
        \rightmargin   0cm
        \parsep        0pt}%
    \item\relax}
{\endlist}

\newenvironment{mycover}
{\list{}{\listparindent 0pt
        \itemindent    \listparindent
        \leftmargin    0cm
        \rightmargin   1.5cm
        \parsep        0pt}%
    \raggedright
    \item\relax}
{\endlist}

\begin{mycover}

\textbf{Alkida Balliu}
\myaff{Gran Sasso Science Institute, Italy}

\textbf{Filippo Casagrande}
\myaff{Gran Sasso Science Institute, Italy}

\textbf{Francesco d'Amore}
\myaff{Gran Sasso Science Institute, Italy}

\textbf{Dennis Olivetti}
\myaff{Gran Sasso Science Institute, Italy}


\bigskip
\end{mycover}

\begin{myabstract}
\noindent\textbf{Abstract.}
    Very recently, Khoury and Schild [FOCS 2025] showed that any randomized LOCAL algorithm that solves maximal matching requires $\Omega(\min\{\log \Delta, \log_\Delta n\})$ rounds, where $n$ is the number of nodes in the graph and $\Delta$ is the maximum degree. This result is shown through a new technique, called \emph{round elimination via self-reduction}.
    The lower bound proof is beautiful and presents very nice ideas. However, it spans more than 25 pages of technical details, and hence it is hard to digest and generalize to other problems.

    Historically, the simplification of proofs and techniques has marked an important turning point in our understanding of the complexity of graph problems. Our paper makes a step forward in this direction, and provides the following contributions.

    \begin{enumerate}
        \item We present a short and simplified version of the round elimination via self-reduction technique. The simplification of this technique enables us to obtain the following two hardness results.
        \item We show that any randomized LOCAL algorithm that solves the maximal $b$-matching problem requires $\Omega(\min\{\log_{1+b}\Delta, \log_\Delta n\})$ and $\Omega(\sqrt{\log_{1+b} n})$ rounds. We recall that the $b$-matching problem is a generalization of the matching problem where each vertex can have up to $b$ incident edges in the matching. As a corollary, for $b=1$, we obtain a short proof for the maximal matching lower bound shown by Khoury and Schild.
        \item We show that any randomized LOCAL algorithm that properly colors the edges of a graph with $\Delta + k$ colors requires $\Omega(\min\{\log \Delta, \log_\Delta n\})$ and $\Omega(\sqrt{\log n})$ rounds, for any $k\le \Delta^{1-\varepsilon}$ and any constant $\varepsilon > 0$.
    \end{enumerate}
\end{myabstract}

\thispagestyle{empty}
\setcounter{page}{0}
\newpage

%% file: trunk/introduction.tex
\section{Introduction}

In this work, we study the hardness of fundamental graph problems in the well-studied \local model of distributed computing. In this model, each node of the input graph represents a computer, and edges represent communication links. Nodes, which communicate through message passing, have unbounded local computational power and the size of the messages that they exchange can be arbitrarily large. 
Since this model is very strong, lower bounds that hold in the \local model are of particular interest: they directly extend to all of the many weaker models studied in the literature. However, proving lower bounds for the \local model is challenging. 

\subsection{Lower bound techniques in the \local model}
The known non-trivial lower bounds for the \local model can be essentially divided in two main groups: those that are based on the \emph{round elimination technique}, and those that explicitly design some hard instances where it is possible to apply some \emph{indistinguishability arguments}.

\paragraph{Round elimination.}
At a high level, lower bounds shown via the round elimination technique consist of constructing a sequence of problems that satisfy some desired properties. More precisely, let $P_0$ be a problem of interest for which we want to show a lower bound. The goal (and the challenge) is to produce a sequence $P_0, P_1, \ldots, P_T$ of problems such that the following two properties hold: (i) for all $0<i\le T$ it holds that $P_i$ is one round easier than $P_{i-1}$; (ii) $P_T$ cannot be solved in $0$ rounds. Such a sequence implies that $P_0$ cannot be solved in $T$ rounds and hence we obtain a lower bound for $P_0$.
The first \local lower bounds based on this technique date back to the late '80s and the beginning of the '90s, and they show that $3$-coloring cycles requires $\Omega(\log^* n)$ rounds (where $n$ is the number of nodes in the graph), both for deterministic \cite{linial-1992-locality-in-distributed-graph-algorithms} and randomized \cite{naor-1991-a-lower-bound-on-probabilistic-algorithms-for} algorithms. 
Even though at a first glance such a technique may look conceptually easy, for many years researchers found it very difficult to use it for showing lower bounds for different graph problems in the \local model.
In fact, the first round elimination lower bound after the aforementioned ones was shown more than 25 years later,
in a paper that presents a lower bound for the sinkless orientation problem,\footnote{A solution to sinkless orientation requires orienting edges such that no node is a sink.} the proof of which is very long and involved \cite{brandt-fischer-etal-2016-a-lower-bound-for-the}.

\paragraph{Indistinguishability.}
Lower bounds shown via indistinguishability arguments consist of constructing a class of graphs that have some desired properties.
At a very high level, this technique consists of showing that there are nodes that, within their running time, cannot distinguish between two different settings in which they would be required to give different outputs, reaching a contradiction.
Perhaps one of the most influential indistinguishability-type lower bounds is the so-called KMW lower bound, shown in 2004, which provides a lower bound for many fundamental graph problems \cite{kuhn-moscibroda-wattenhofer-2004-what-cannot-be}. The KMW construction is very nice and clever. However, the proof is quite complex. 
For more than 15 years, the KMW construction was one of the few (non-trivial) indistinguishability-type lower bounds, and in fact, until recently, there were no new constructions inspired by the KMW one. Arguably, the involved proof of the KMW lower bound may have been an obstacle.

\subsection{Simplifying proofs and techniques marked a turning point}
The deeper understanding and simplification of the existing lower bound techniques and proofs marked an essential turning point in our understanding of the hardness of several graph problems in the \local model.

\paragraph{Automatic round elimination.}
Regarding lower bounds proved via the round elimination technique, the turning point happened in 2019, when the \emph{automatic round elimination} technique was introduced \cite{brandt-2019-an-automatic-speedup-theorem-for}. Before this work, we had to \emph{guess} what could be the sequence $P_0, P_1, \ldots, P_T$ of problems that have the aforementioned desired properties and that hence implies a lower bound for $P_0$. The automatic round elimination technique greatly simplifies this part, as it is able to \emph{automatically construct such a sequence}. 
While it does not make it trivial to prove all the lower bounds that we care about (e.g., we still do not have lower bounds, as a function of $\Delta$, for $(\Delta+1)$-vertex coloring), this technique has been used to prove many interesting lower bounds for fundamental graph problems (see, e.g., \cite{balliu-brandt-etal-2021-lower-bounds-for-maximal,brandt-olivetti-2020-truly-tight-in-delta-bounds-for,balliu-brandt-olivetti-2022-distributed-lower-bounds,balliu-brandt-etal-2021-improved-distributed-lower,balliu-brandt-etal-2022-distributed-delta-coloring,balliu-brandt-etal-2020-classification-of-distributed,balliu-brandt-etal-2023-distributed-maximal-matching,balliu-brandt-etal-2025-distributed-quantum-advantage}).
For example, the automatic version of round elimination made it possible to solve a long-standing open question regarding 
the hardness of the maximal matching problem:\footnote{A matching in a graph $G$ is a selection of node-disjoint edges, and a matching in $G$ is maximal if it is not a proper subset of any other matching in $G$.} 
indeed, \cite{balliu-brandt-etal-2019-lower-bounds-for-maximal} showed that any randomized \local algorithm that solves maximal matching requires $\Omega(\min\{\Delta, \log_\Delta\log n\})$ rounds.
Moreover, this technique heavily simplified existing proofs: using automatic round elimination, the hardness of sinkless orientation follows from a two-line proof.

\paragraph{A breezing proof for KMW.}
The KMW construction uses indistinguishability-type arguments to show lower bounds for several fundamental graph problems. For example, it shows that any randomized \local algorithm that solves maximal matching requires $\Omega(\min\{\log \Delta/\log\log\Delta, \log_\Delta n\})$ rounds, even on trees. 
Since a maximal matching on some graph $G$ is a maximal independent set on the line graph of $G$,\footnote{An independent set in a graph $G$ is a selection of pairwise non-adjacent nodes, and an independent set in $G$ is maximal if it is not a proper subset of any other independent set in $G$. The line graph $G'$ of a graph $G$ is constructed as follows: for any edge of $G$ we have a node in $G'$ and there is an edge between two nodes in $G'$ if and only if the edges in $G$ that they represent share a node.} 
the KMW lower bound directly extends to maximal independent set as well, 
but in this specific case the lower bound does not hold on trees.
Perhaps the turning point that helped to better understand the KMW lower bound happened in 2020, when \cite{coupette-lenzen-2021-a-breezing-proof-of-the-kmw-bound} presented a breezing proof of the KMW lower bound. This paper made the KMW construction more accessible, and it helped in proving new lower bounds inspired by KMW. Indeed, shortly after, \cite{balliu-ghaffari-etal-2022-node-and-edge-averaged} came up with a novel adaptation of KMW, showing, among other things, a lower bound for maximal independent set on trees, resolving the complexity of maximal independent set on trees up to an $O(\sqrt{\log\log n})$ factor and closing an open question stated in \cite{barenboim-elkin-2013-distributed-graph-coloring}.

\subsection{Round elimination via self-reduction.}
Very recently, \cite{khoury2025} showed that any \local randomized algorithm that solves maximal matching requires $\Omega(\min \{\log \Delta, \log_\Delta n\})$ rounds in the \local model, improving the KMW lower bound (see \cref{table:matching-lb} for a summary on the randomized lower bounds for maximal matching). 
This result was obtained through a novel application of the round elimination technique, called \emph{round elimination via self-reduction}.

While the automatic round elimination technique of \cite{brandt-2019-an-automatic-speedup-theorem-for} works for so-called \emph{locally checkable problems}, this new technique works by constructing a sequence of \emph{optimization problems}. 
In more detail, the authors of \cite{khoury2025} constructed a sequence 
$P_0,\ldots,P_T$, where each problem $P_i$ is the problem of approximating a maximum matching, but where the approximation factor decreases at each step.
In other words, the authors show that, if we can solve in $T$ rounds a matching where the unmatched nodes are not too many, then in one round less we can produce a matching that still has a large-enough size. 
By repeating this reasoning iteratively, the authors get their lower bound result. 
This is a novel and very nice idea. However, unfortunately, this new lower bound is involved and its proof for maximal matching spans more than 25 pages. This makes the round elimination via self-reduction technique hard to digest and generalize to other graph problems.

\begin{table}[h]
    \begin{tabular}{lllll}
    & Technique &Randomized lower bound  &Solely as a function of $n$ & Paper\\
    \hline
    & KMW & $\Omega(\min\{\log \Delta/\log\log\Delta, \log_\Delta n\})$  & $\Omega(\sqrt{\log n/ \log\log n})$ & \cite{kuhn-moscibroda-wattenhofer-2004-what-cannot-be}\\
    & Automatic RE &$\Omega(\min\{\Delta, \log_\Delta\log n\})$  & $\Omega(\log\log n / \log \log \log n)$ & \cite{balliu-brandt-etal-2019-lower-bounds-for-maximal}\\
    & RE via self-reduction &$\Omega(\min\{\log \Delta, \log_\Delta n\})$  &  $\Omega(\sqrt{\log n})$ & \cite{khoury2025}
    \end{tabular}
    \caption{The table summarizes the lower bounds for maximal matching obtained over the years. These lower bounds depend both on the maximum degree $\Delta$ and on the number $n$ of nodes of the graph. The lower bound of \cite{balliu-brandt-etal-2019-lower-bounds-for-maximal} is not comparable with the other two lower bounds: compared to the other two, this lower bound is stronger as a function of $\Delta$ and weaker as a function of $n$. The work of \cite{khoury2025} strictly improves the KMW lower bound.}
    \label{table:matching-lb}
\end{table}

\subsection{Our Contributions}
The contribution of this paper is twofold: we introduce a greatly simplified round elimination via self-reduction technique, and we apply it to show lower bounds for some fundamental graph problems. 

\paragraph{A simplified round elimination via self-reduction technique.}
We present a simple and clean round elimination via self-reduction technique. 
As a consequence, we provide a much shorter and simpler proof for obtaining the same lower bound as \cite{khoury2025} for maximal matching. 
Our round elimination via self-reduction technique is different from that of \cite{khoury2025}, and it avoids many new concepts that are used and needed in \cite{khoury2025}, resulting in an easier-to-understand technique. We believe that this will mark a turning point for understanding the hardness of other graph problems in the \local model using the round elimination via self-reduction technique.
On a high level, the authors of \cite{khoury2025} proved a lower bound for maximal matching by considering the sequence of problems $P_i$ of finding a matching where at most a fraction $p_i$ of nodes are unmatched, in a model where the computational entities are the \emph{edges} of the graph.
We show that, by considering a different problem, which we call $1$-grabbing, and the more standard setting in which algorithms run on \emph{nodes}, the proof becomes much simpler. This problem requires each node to select exactly one of its incident edges, and the goal is to minimize the edges grabbed by only one endpoint.
In the setting considered by \cite{khoury2025}, the definition of the $T-1$ round algorithm requires the concept of \emph{$\delta$-good flowers}, which complicates the definition of the algorithm and the proofs. We show that, by considering the $1$-grabbing problem, and by considering \emph{node} algorithms, we can avoid this concept entirely.

\paragraph{A lower bound for maximal $b$-matching.}
We do not explicitly provide a proof for $1$-grabbing. Instead, we consider a more general problem called $b$-grabbing. This allows us to prove a lower bound for a well-studied generalization of the maximal matching problem, that is called maximal $b$-matching.
In this problem, the goal is to select a maximal set of edges such that each node has at most $b$ incident selected edges. 
We say that a $b$-matching is maximal if it is not a proper subset of any other $b$-matching (see \cref{def:prelim:b-matching} for a formal definition). The standard maximal matching problem is obtained by the special case where $b=1$. 
We show that any \local randomized algorithm that solves maximal $b$-matching requires $\Omega(\min\{\log_{1+b} \Delta, \log_\Delta n\})$ and $\Omega(\sqrt{\log_{1+b} n})$ rounds, even on trees, and even when nodes have access to shared randomness (see \Cref{theorem:analysis:lower-bound-relaxed-b-matching}). 

\paragraph{A lower bound for $(\Delta + k)$-edge coloring.}
Finally, we show that our proofs easily extend to a lower bound for the $(\Delta + k)$-edge coloring problem, showing that, for any $k\le \Delta^{1-\varepsilon}$ and constant $\varepsilon > 0$, any randomized \local algorithm that solves $(\Delta + k)$-edge coloring requires 
$\Omega(\min\{\log \Delta, \log_\Delta n\})$ and $\Omega(\sqrt{\log n})$ rounds, even on trees, and even when nodes have access to shared randomness (see \Cref{theorem:analysis:edge-coloring-lower-bound}). 
This is a generalization of the well-known $(\Delta + 1)$-edge coloring problem (Vizing's theorem). 
For this problem, prior to our result, the best lower bound was inherited from the $(2\Delta - 2)$-edge coloring problem studied in \cite{chang-he-etal-2018-the-complexity-of-distributed-edge}, which requires $\Omega(\log_\Delta \log n)$ rounds w.h.p. and $\Omega(\log_\Delta n)$ rounds deterministically.
Among the known upper bounds for the edge-coloring problem that we consider \cite{christiansen2023,Bernshteyn2022,BERNSHTEYN202569}, to the best of our knowledge, the best upper bound for $(\Delta + k)$-edge coloring, for large values of \(\Delta\), comes from the $O(\poly(\Delta) \poly(\log n)) $-round algorithm for \((\Delta+1)\)-edge coloring obtained in \cite{BERNSHTEYN202569}.

%% file: trunk/preliminaries.tex
\section{Preliminaries}

In this section, we introduce some notation, some definitions, and the model of computing that we consider.

\paragraph{Graphs.}
We work with simple undirected graphs $G = (V, E)$, where $V$ is the set of vertices and $E$ is the set of edges.
If \(V\) and \(E\) are not specified, we denote by \(V(G)\) the set of nodes of \(G\) and by \(E(G)\) the set of edges of \(G\). 
We say that two nodes \(u,v\) are \emph{adjacent} if there is an edge \(\{u,v\} \in E\) between them. 
The degree of a node \(v\), denoted by \(\deg_G(v)\), is the number of nodes that are adjacent to \(v\) in \(G\).
The maximum degree of a graph \(G\), denoted by \(\maxDeg\), is the maximum degree over all nodes in \(G\).

Given a graph \(G = (V,E)\) and any two nodes \(u,v \in V\), the distance between \(u\) and \(v\), denoted by \(\dist_G(u,v)\), is the length of a shortest path between \(u\) and \(v\) in \(G\).
We omit the subscript \(G\) when it is clear from the context.
The diameter of \(G\) is denoted by \(\diam(G)\) and is defined as the maximum distance between any two nodes in \(G\).
For any graph \(G = (V,E)\) and any node \(v \in V\), we denote by \(\NN_r[v]\) the radius-\(r\) closed neighborhood of \(v\), i.e., the set of vertices that are at distance at most \(r\) from \(v\) in \(G\).
We can easily extend this notation to subsets of nodes \(S \subseteq V\) by defining \(\NN_r[S] = \bigcup_{v \in S} \NN_r[v]\).

For any subset of nodes \(S \subseteq V\), we denote by \(G[S]\) the subgraph of \(G\) induced by \(S\), i.e., the graph whose vertex set is \(S\) and whose edge set contains all edges in \(E\) that have both endpoints in \(S\).
The girth of a graph \(G\), denoted by \(\girth(G)\), is the length of a shortest cycle in \(G\).
If \(G\) is acyclic, then we say that \(G\) has infinite girth.

\paragraph{Half-edges and labelings.}
Given any node \(v \in V\) and any edge \(e \in E\) that is incident to \(v\), the node-edge pair \((v,e)\) is called \emph{half-edge}. 
We also use the short notation \(ve = (v,e)\).
We say that the half-edge \((v,e)\) is \emph{incident} to \(v\).
Note that any edge \(e = \{u,v\}\) is uniquely split into two half-edges \((u,\{u,v\})\) and \((v,\{u,v\})\).
Let \(H(G)\) be the set of all half-edges of $G$.
An \emph{edge-labeled} graph \(G\) is a pair \((G = (V,E), \ell)\) where \(\ell\colon H(G) \to \Sigma\) is a function that assigns a label from a set of labels \(\Sigma\) to all half-edges of \(G\): \(\ell\) is the \emph{labeling} of \(G\).

\paragraph{Independence number.}
An independent set of a graph \(G\) is a subset of nodes \(I \subseteq V\) such that no two nodes in \(I\) are adjacent.
The independence number of \(G\), denoted by \(\indNum(G)\), is the size of a maximum independent set in \(G\).

\paragraph{Isomorphic graphs.}
Given two graphs \(G\) and \(H\), we say that \(G\) is \emph{topologically isomorphic}\footnote{We specify the term ``topologically'' because we will use a more restrictive notion of isomorphism later in \cref{preliminar:isomorphism}.}
to \(H\) if there exists a bijection \(\varphi: V(G) \to V(H)\) such that \(\{u,v\} \in E(G)\) if and only if \(\{\varphi(u), \varphi(v)\} \in E(H)\).
Such \(\varphi\) is called a topological isomorphism between \(G\) and \(H\).

\paragraph{Probability notions.}
For an event \(\EE\), we write \(\EE^C\) to denote the complementary event of \(\EE\), and we write \(\mathds{1}_\EE\) to denote the indicator random variable of \(\EE\), i.e., \(\mathds{1}_\EE = 1\) if \(\EE\) holds, and \(\mathds{1}_\EE = 0\) otherwise.

\paragraph{The \local model.}

We work in the \local model of distributed computing~\cite{linial-1992-locality-in-distributed-graph-algorithms}.
In this model, we are given a distributed system of \(n\) processors/nodes that can communicate over a communication network.
The network is represented as a graph \(G = (V,E)\), where each node in \(V\) corresponds to a processor and each edge in \(E\) corresponds to a communication link between two processors.
At the beginning of the computation, all nodes are identical, except for a unique identifier that is assigned to each node from the set \(\{1, \ldots, n^c\}\) for some fixed constant \(c \geq 1\), and for a local variable that stores input data.
Input data consists of the input given by the considered graph problem, the degree of the node, and port numbers that allow the node to distinguish its incident communication links (note that port numbers are assigned arbitrarily by an adversary).
We also assume that nodes know the number \(n\) of processors in the graph.
The computation proceeds in synchronous rounds and every node runs the same algorithm:
in each round, a node can perform arbitrary local computation and send a message of arbitrarily large size to each of its neighbors. 
Computation ends when all nodes have terminated running the algorithm and output some local output.
In the setting that we consider, the local output of a node is a label for each incident half-edge, which can be stored locally using port numbers. In other words, local outputs are announced on half-edges.
The running time of an algorithm is the number of communication rounds until computation ends for all nodes.

In the \emph{randomized} \local model, nodes are given access to independent sources of randomness (each node has its own independent infinite stream of random bits).
In the \local model \emph{with shared randomness}, nodes also have access to a shared source of randomness, which is an infinite stream of random bits that is the same for all the nodes and is independent of the private randomness of each node mentioned before.
In the \emph{deterministic} \local model, nodes do not have access to any source of randomness.

Randomized algorithms are allowed to \emph{fail} (i.e., to not produce a correct solution for the given problem) with probability at most $n^{-c}$ for some constant $c \ge 1$. We say that an algorithm is \emph{no-error} if it always produces a valid solution for the given problem.

The \emph{complexity} of a graph problem is the minimum asymptotic running time of an algorithm that solves the problem.
Note that, since nodes do know \(n\), any randomized algorithm can produce a unique identifier for all nodes in the set \([n^c]\), for some constant \(c > 1\), with probability at least \(1 - n^{-(c-1)}\) by letting each node pick u.a.r.\ an identifier from the set \([n^c]\).
We will assume that identifiers are assigned u.a.r.\ when considering randomized algorithms. 
We will also assume that nodes recompute their port numbers locally with a random permutation.
In this work, we will consider the randomized \local model with shared randomness, which is the strongest model, and prove a lower bound in this model: this implies a lower bound in all weaker models.

\paragraph{View of a node.}
With the definition of the model at hand, we can now define the notion of the radius-\(T\) view of a node \(v\).
\begin{definition}[View of a node]\label{def:prelim:view}
    Suppose we are given a distributed system whose network graph is \(G = (V,E)\).
    The radius-\(T\) view of a node \(v \in V\) in \(G\), denoted by \(\view_T(v,G)\), is a distributed system whose network is represented by the graph \(H\) with \(V(H) = \neighborhood_T[v]\) and such that \(\{u,w\} \in E(H)\) if and only if the following holds: \(\dist_G(u,v) \le T\) and \(\dist_G(u,w) \le T-1\), or vice-versa.
    Each node in \(\view_T(v,G)\) has the same input as in the input distributed system (hence, identifier, port numbers, and -- possibly -- randomness).
\end{definition}

\paragraph{Equivalence between graphs and distributed systems.}\label{preliminar:isomorphism}
From now on, for simplicity, when referring to an input graph \(G = (V,E)\) for a problem in the \local model, we consider \(G\) to be the whole distributed system, so that nodes of \(G\) store identifiers, port numbers, and possibly input data and randomness.
Given two input graphs \(G\) and \(H\), we say that \(G\) is isomorphic to \(H\) if \(G\) and \(H\) are topologically isomorphic, and the topological isomorphism \(\varphi\) preserves identifiers, port numbers, and input data (and randomness, if any).
In such a case, we say that \(\varphi\) is an isomorphism between the input graphs \(G\) and \(H\).
Such a definition also extends to views of nodes \(\view_T(v,G)\) and \(\view_T(w,H)\): in this case we also ask that the isomorphism brings \(v\) to \(w\).

\paragraph{Maximal \(b\)-matching.}
In this paper, we consider the problem of finding a maximal \(b\)-matching in a graph, for any \(b \in \natsPos\) (\(b\) can also be a function of the input graph parameters, such as the maximum degree \(\maxDeg\)).

\begin{definition}[Maximal \(b\)-matching]
\label{def:prelim:b-matching}
Let \(G = (V, E)\) be a graph.
A \emph{\(b\)-matching} of \(G\) is a subset of edges \(M \subseteq E\) such that every node \(v \in V\) is adjacent to at most \(b\) edges in \(M\).
A \(b\)-matching \(M\) is \emph{maximal} if there is no edge \(e \in E \setminus M\) such that \(M \cup \{e\}\) is still a \(b\)-matching.
We say that a node \(v \in V\) is \emph{maximal} if no edge that is incident to \(v\) can be added to the \(b\)-matching without violating the \(b\)-matching property.
Note that by taking \(b = 1\) we obtain the classical definition of maximal matching.
We say that an edge is \emph{matched} if and only if it is in $M$.
\end{definition}

\paragraph{Regular graphs.}
A graph is $\Delta$-regular if all nodes have degree $\Delta$. A graph is regular if it is $\Delta$-regular for some $\Delta$.
We will prove a lower bound that holds already in the case of regular graphs. Hence, in the rest of the paper, we restrict ourselves to the case of $\Delta$-regular graphs.

\paragraph{Graphs with low independence number.}
We will use the fact that there exist regular graphs with high girth and low independence number. Such graph families have been shown to exist in \cite{bollobas80,bollobas88,brandt-chang-etal-2022-local-problems-on-trees-from-the}. We use a graph family that has already been considered in previous works (see, e.g., \cite{balliu-boudier-etal-2024-tight-lower-bounds-in-the}).

\begin{lemma}[Lemma 2.1 of \cite{balliu-boudier-etal-2024-tight-lower-bounds-in-the}]\label{lemma:prelim:lb-graph}
    There exist two positive constants \(\rho, \varepsilon\) such that the following holds.
    Let \(n,\maxDeg\) be two positive integers such that \(n \maxDeg\) is even, and such that \(2 \le \maxDeg \le n\).
    There exists a graph \(G_{n,\maxDeg}\) of \(n\) nodes with the following properties:
    \begin{itemize}
        \item (Regularity) \(G_{n,\maxDeg}\) is \(\maxDeg\)-regular.
        \item (Girth) \(G_{n,\maxDeg}\) has girth at least \(\varepsilon \log_\maxDeg (n)\).
        \item (Independence number) \(\indNum(G_{n,\maxDeg}) \le \rho n \log (\maxDeg) / \maxDeg\).
    \end{itemize}
\end{lemma}

%% file: trunk/analysis.tex
\section{\texorpdfstring{Lower bound for maximal \(b\)-matching}{Lower bound for maximal b-matching}}\label{sec:matching-lb}

Throughout this section, all algorithms we consider run in the randomized \local model with shared randomness, unless otherwise specified.

\subsection{Overview of our proof}
In order to prove our lower bound, we first need to introduce a different problem, called $b$-grabbing.
\begin{definition}[$b$-grabbing]
    Let $b$ be a positive integer.
    In the $b$-grabbing problem, each node needs to output $\M$ on exactly $b$ incident half-edges, and $\U$ on all the other incident half-edges.
    Half-edges labeled $\M$ by $v$ are said to be \emph{grabbed} by $v$.
    If an edge $e= \{u,v\}$ satisfies that $ve$ is grabbed by $v$ and $ue$ is grabbed by $u$, then we say that $e$ is \emph{matched}.
\end{definition}

\begin{definition}[Saturated node]
    Consider any solution to the \(b\)-grabbing problem or to the maximal \(b\)-matching problem on some graph \(G = (V, E)\).
    A node \(v \in V\) is \emph{saturated} if there exist exactly \(b\) matched edges incident to $v$, otherwise it is \emph{unsaturated}.
\end{definition}

\begin{definition}[Quality and badness]
    The \emph{quality} $q$ of a solution for $b$-grabbing is the number of matched edges.
    In an $n$-node graph, the parameter $p = 1 - \frac{2q}{bn}$ is called \emph{badness}. Observe that  $q = b \frac{n}{2} (1 - p)$.
\end{definition}
\noindent Note that $q \in [0, {bn}/{2}]$, hence $p \in [0,1]$. Moreover, observe that the badness of a solution for a $b$-grabbing instance is $0$ if and only if every node is saturated, while it is $1$ if and only if no edge is matched. So, as the word suggests, we can interpret the badness as a measure of how bad a solution for $b$-grabbing is in terms of representing a maximum $b$-matching.

Let $\mathcal{G}$ be the family of graphs given by \Cref{lemma:prelim:lb-graph}.
On a high level, our lower bound for maximal $b$-matching is obtained as follows.
\begin{enumerate}
    \item We prove that, in $\mathcal{G}$, an algorithm for maximal $b$-matching implies an algorithm (with the same runtime) for $b$-grabbing with expected badness at most $p = \Theta(b \frac{\log \maxDeg}{\maxDeg})$ (see \Cref{ssec:relation}). 
    \item Then, we prove that an algorithm for $b$-grabbing with runtime $T$ and expected badness $p_i$ can be converted into an algorithm for $b$-grabbing with runtime $T-1$ and expected badness at most $p_{i+1} = c\cdot \sqrt{b} \cdot p_i$ for some universal constant $c$ (see \Cref{ssec:onestep}).
    \item We show that any $0$-round algorithm for $b$-grabbing must have expected badness at least some positive constant independent of $\maxDeg$ (see \Cref{ssec:zero}).
    \item We iterate point 2 recursively, as long as we do not contradict the bound of point 3. We show that we can iterate $\Omega(\log_b \maxDeg)$ times. The number of iterations is a lower bound on the runtime of any maximal $b$-matching algorithm (see \Cref{ssec:iteration}).
    \item We show that a lower bound for $\mathcal{G}$ implies a lower bound for trees as well (see \Cref{ssec:trees}).
\end{enumerate}
In the following, $\varepsilon$ is the value given by \Cref{lemma:prelim:lb-graph}.
\subsection{Relation between maximal \texorpdfstring{$b$}{b}-matching and \texorpdfstring{$b$}{b}-grabbing}\label{ssec:relation}
Recall that we restrict our attention to $\maxDeg$-regular graphs.
We now prove that any maximal \(b\)-matching $T$-round algorithm can be converted into a $b$-grabbing $T$-round algorithm with badness that can be upper bounded using the independence number of the input graph.
We assume that we are given an algorithm $\algo$ for maximal $b$-matching. The output of $\algo$ on node $v$ is a subset $P_v$ of $\{1,\ldots,\maxDeg\}$ that indicates which edges incident to $v$ are part of the matching. We consider algorithms that may \emph{fail} with some probability $\bar{p}$. In this context, there are three possible ways to fail:
\begin{itemize}
    \item There exists some node $v$ for which the size of $P_v$ is strictly larger than $b$.
    \item For two neighboring nodes $u$ and $v$, such that $u$ is the $i$th neighbor of $v$ and $v$ is the $j$th neighbor of $u$, it does not hold that $i \in P_v \iff j \in P_u$, that is, $u$ and $v$ do not \emph{agree} on whether the edge $\{u,v\}$ is in the matching or not. An edge is said to be in the matching, or agreed to be in the matching, if both endpoints agree on it being in the matching.
    \item The produced $b$-matching is not maximal.
\end{itemize}

\begin{lemma}\label{lemma:analysis:matching-certified-reduction}
    Any \(T\)-round randomized \local algorithm \(\algo\) for maximal \(b\)-matching with failure probability $\bar{p}$ can be transformed into a no-error \(T\)-round randomized \local algorithm \(\algo'\) for \(b\)-grabbing with expected badness at most \(b \indNum(G)/n + \bar{p}\), where $\indNum(G)$ is the independence number of the input graph $G$. 
\end{lemma} 
\begin{proof}
    The algorithm \(\algo'\) is defined as follows.
    \begin{enumerate}
        \item Each node $v$ spends $T$ rounds to run $\algo$. After this operation, $v$ knows $P_v$. Recall that, with probability at most $\bar{p}$, $\algo$ may have failed on at least one node. There are three possible failure modes: a node may have more than $b$ incident edges in the matching, nodes may not agree on which edges are in the matching, or the $b$-matching is not maximal.
        \item Let $m_v = |P_v|$. If $v$ has more than $b$ incident edges in the matching, node $v$ updates $P_v$ by picking arbitrary $m_v -b$ ports and removing them from $P_v$. Observe that, with probability at least $1-\bar{p}$, no set $P_v$ changed after this operation, and for all edges $e = \{u,v\}$ the nodes $u$ and $v$ agree on whether $e$ is in the matching or not.
        \item Each node $v$ now satisfies $m_v \le b$. Node $v$ outputs $\M$ on each port in $P_v$. Then, it picks other $b - m_v$ ports at random and outputs $\M$ on them. Finally, it outputs $\U$ on all other incident half-edges.
    \end{enumerate}
    Note that, by construction, \(\algo'\) is a \(T\)-round no-error randomized \local algorithm that solves \(b\)-grabbing. We now bound the expected badness of \(\algo'\).
    We consider two cases separately.
    \begin{itemize}
        \item If $\algo$ failed, we can upper bound the badness of $\algo'$ by $1$.
        \item If $\algo$ did not fail, we know that the $b$-matching $M$ produced by $\algo$ is correct and maximal. Let $U$ be the set of unsaturated nodes produced by $\algo$. Consider the subgraph $G' = G[U]$ induced by unsaturated nodes. Observe that each edge of $G'$ must be in $M$, since otherwise $M$ would not be maximal. This implies that $G'$ has maximum degree $b-1$, and hence it can be colored with $b$ colors. Since each color class must be an independent set in $G$, which has independence number $\alpha(G)$, we obtain that the number of nodes of $G'$, and hence the number of unsaturated nodes in $G$, is at most $b \alpha(G)$. Thus, the number of matched edges can be lower bounded by $b(n - b \alpha(G))/2 = b\frac{n}{2}(1 - b \alpha(G)/ n)$. Hence, the badness is upper bounded by $b \alpha(G) / n$.
    \end{itemize}
    By combining the two bounds, we obtain that the expected badness is upper bounded by \[
        1 \cdot \bar{p} + b \alpha(G) / n \cdot (1-\bar{p}) \le  \bar{p} + b \alpha(G) / n
    \]
\end{proof}

\begin{corollary}\label{cor:starting}
    There exists a no-error $T$-round algorithm $\algo'$ that solves $b$-grabbing with expected badness at most $\Theta(b \frac{\log \maxDeg}{\maxDeg})$ in any $G \in \mathcal{G}$, where $T$ is the complexity of the best randomized algorithm for maximal $b$-matching on $\mathcal{G}$ with failure probability at most $b \frac{\log \maxDeg}{\maxDeg}$.
\end{corollary}
\begin{proof}
    Let $\algo$ be an algorithm that solves maximal $b$-matching with failure probability at most $b \frac{\log \maxDeg}{\maxDeg}$ in any $G \in \mathcal{G}$.
    By \Cref{lemma:prelim:lb-graph}, $\alpha(G) \le \rho n \log (\maxDeg) / \maxDeg$. By applying \Cref{lemma:analysis:matching-certified-reduction}, we obtain a no-error $b$-grabbing algorithm with expected badness upper bounded by $\Theta(b \frac{\log \maxDeg}{\maxDeg} + b \frac{\log \maxDeg}{\maxDeg})$. Hence, the claim follows.
\end{proof}

\subsection{Round elimination via self-reduction}\label{sec:analysis:round-elimination}\label{ssec:onestep}
We devote this subsection to proving the following statement.
\begin{lemma}\label{lem:onestep}
    Let $T \le \frac{\varepsilon}{4} \log_\maxDeg n$.
    Let $\algo$ be a no-error $T$-round algorithm that solves $b$-grabbing with expected badness upper bounded by $p$. Then, there exists a no-error $(T-1)$-round algorithm $\algo'$ that solves  $b$-grabbing with expected badness upper bounded by $c \cdot \sqrt{b} \cdot p$, for some universal constant $c > 1$.
\end{lemma}
In the following, let $p_0 = p$, $p_1 = c \cdot \sqrt{b} \cdot p$, $\algo_0 = \algo$, $\algo_1 = \algo'$, and let \(G\) be any input graph.

\paragraph{Some notation.}
For a random variable \(X\) whose probability space is given by the private random bit assignments to the nodes of $G$ and the shared randomness assigned to the nodes, we denote by \(\expect{X}\) the expectation of $X$, and by \(\expectView{T}{X}\) the expectation of $X$ conditioned on the fact that the radius-\(T\) view of the node \(v\) is isomorphic to \(\view_{T}(v,G)\).

For all \(x,y \in \{\M,\U\}\) and all \(i \in \{0,1\}\), we define \(E^{(i)}_{xy}\) and \(H^{(i)}_{x}\) to be the number of edges that are labeled \(xy\) and the number of half-edges that are labeled \(x\), respectively, by \(\algo_i\).

When a node \(v\) is specified, we denote by \(E^{(i)}_{xy}(v)\) and \(H^{(i)}_{x}(v)\) the number of edges incident to \(v\) that are labeled \(xy\) or \(yx\) and the number of half-edges incident to \(v\) that are labeled \(x\), respectively, by \(\algo_i\).

\paragraph{Preferred directions and edges.}
Let $v \in V(G)$. For each $i \in [\maxDeg]$, let $X_i(v)$ be the indicator random variable that is equal to $1$ if $\algo_0$ outputs $\M$ on the $i$th port of $v$, and $0$ otherwise. 
Let $x_i(v) = \expectView{T-1}{ X_i(v) }$. In other words, $x_i(v)$ is the fraction of extensions of the radius-$(T-1)$ view of $v$ into radius-$T$ views, such that $\algo_0$ outputs $\M$ on the $i$th port of $v$.

The \(b\) indices \(i_1, \ldots, i_b\), such that \(x_{i_1}(v), \ldots, x_{i_b}(v)\) are the largest among \(x_1(v), \ldots, x_{\maxDeg}(v)\), are called the \emph{preferred directions} of \(v\), with ties broken by giving priority to the smallest subscripts. 
An edge (resp.\ half-edge) is a \emph{preferred edge} (resp.\ \emph{preferred half-edge}) for $v$ if it is incident on a port that is a preferred direction.
    We furthermore define \(S(v)\) to be the sum of \(x_1(v),\ldots,x_{\maxDeg}(v)\), we define \(S_i(v)\) to be the sum of the \(i\) largest values among \(x_1(v),\ldots,x_{\maxDeg}(v)\), and \(S_{\maxDeg,i}(v) = S(v) - S_i (v)\).

By linearity of expectation, and the fact that $\algo_0$ outputs $\M$ exactly $b$ times on the ports of $v$, we obtain the following.

\begin{observation}
    Let \(G\) be any input graph.
    For each node \(v \in V\), it holds that \(S(v) = b\).
\end{observation}

\paragraph{Definition of $\algo_1$.}
We define $\algo_1$ as follows.
Fix the input graph \(G\), and fix any node \(v \in V\).
The node \(v\) gathers its radius-\((T-1)\) view \(\view_{T-1}(v,G)\), computes its \(b\) preferred directions, outputs $\M$ on the preferred half-edges, and $\U$ on all the others.

In other words, each node $v$ gathers its radius-$(T-1)$ view $\mathcal{V}$. Then, it considers all possible extensions of $\mathcal{V}$ into radius-$T$ views, it simulates $\algo_0$ on all of them, and for each incident half-edge $h$ it computes the fraction of times in which $\algo_0$ outputs $\M$ on $h$.
Then, it outputs $\M$ on the $b$ ports with the highest fractions, and $\U$ on all the others.

\paragraph{Almost determinism.}
We now prove that, in some sense, the algorithm $\algo_0$ must be almost deterministic.
Informally, $S_{\maxDeg,b}(v)$ is the fraction of times in which $\algo_0$ outputs $\M$ on a direction that is not preferred.
We prove that, if the sum of $S_{\maxDeg,b}(v)$ is large, then $\algo_0$ must have large badness. 
For this purpose, in the full version of this paper \cite{full_version_arxiv}, we introduce a statement about probability; for the case in which $b=1$, we also give a much simpler proof.

\begin{lemma}\label{lem:probability-deviation}
    Let $x_1, \ldots, x_\maxDeg$. 
    Assume that $\sum{x_i} = b$.
    Let $Y_1,\ldots,Y_\maxDeg$ be mutually independent events, where $Y_i$ has probability $y_i$. 
    Define $S_b$ to be the sum of the largest $b$ values in $x_1,\ldots,x_\maxDeg$, and let $S_{\maxDeg,b} = b - S_b$.
    Let $Y$ be the number of events that occur among $Y_1,\ldots,Y_\maxDeg$. Assume that $\sum_{i \in [\maxDeg]}\abs{x_i - y_i} < S_{\maxDeg,b}/(1000 \sqrt{b})$.
    Then, $\expect{\abs{b - Y}} \ge S_{\maxDeg,b} / (1000 \sqrt{b})$.
\end{lemma}

We now prove that the sum of $S_{\maxDeg,b}(v)$ over all nodes $v$ cannot be too large.
\begin{lemma}\label{lemma:analysis:algorithm-A0-expected-MU}
    Fix any input graph \(G\) and any node \(v \in V\).
    It holds that 
    \[
        \expectView{T-1}{E^{(0)}_{\M\U}(v)} \ge S_{\maxDeg,b}(v)/(1000\sqrt{b}).
    \]
\end{lemma}
\begin{proof}
    Let \(v_1, \ldots, v_\maxDeg\) be the \(\maxDeg\) neighbors of \(v\), following the ordering given by the port numbers of \(v\).
    Accordingly, let \(e_1 = \{v,v_1\}, \ldots, e_\maxDeg = \{v,v_{\maxDeg}\}\).
    Let $Y_i$ be the event where $\algo_0$ outputs $\M$ on the port of $v_i$ incident to $e_i$, when conditioned on the $(T-1)$-view of $v$ being $\view_{T-1}(v,G)$, and 
    let $X_i$ be the event where $\algo_0$ outputs $\M$ on the port of $v$ incident to $e_i$, when conditioned on the $(T-1)$-view of $v$ being $\view_{T-1}(v,G)$. 
    Since the girth of $G$ is at least $\varepsilon \log_\Delta n$, and by assumption $T \le \frac{\varepsilon}{4} \log_\maxDeg n$, we get that the events $Y_i$ are mutually independent.
    This holds even if we assume shared randomness. 
    Note that the intersection of the $T$-views of $v_i$ and $v_j$ with $i \neq j$ is exactly the $(T-1)$-view of $v$, because of the value of $T$ with respect to the girth of $G$.
    Moreover, note that \(X_i = (X_i(v) \mid \view_{T-1}(v,G))\), that is, the random variable \(X_i(v)\) conditioned on the $(T-1)$-view of $v$ being $\view_{T-1}(v,G)$.
    Let $y_i = \expect{ \mathds{1}_{Y_i} }$, and \(x_i = \expect{ \mathds{1}_{X_i} }\).
    In other words, $y_i$ is the fraction of times in which $v_i$ outputs $\M$ towards $v$ conditioned on the fixed $(T-1)$-view of $v$, while \(x_i = x_i(v)\) as defined in the paragraph dedicated to the preferred directions of \(v\).
    Let $Y = \sum \mathds{1}_{Y_i}$. 
    Observe that $Y$ is the number of $\M$ labels that the neighbors of $v$ output towards $v$, conditioned on the fixed $(T-1)$-view of $v$. Note that the number of edges incident to $v$ labeled $\M \U$ or $\U \M$ by $\algo_0$ is exactly $\sum_{i \in [\maxDeg]}\abs{\mathds{1}_{X_i} - \mathds{1}_{Y_i}}$. 
    By linearity of expectation, we have that $\expectView{T-1}{E^{(0)}_{\M\U}(v)} = \sum_{i \in [\maxDeg]} \expect{\abs{\mathds{1}_{X_i} - \mathds{1}_{Y_i}}}$.
    By Jensen's inequality, we obtain that \(\expectView{T-1}{E^{(0)}_{\M\U}(v)} \ge \sum_{i \in [\maxDeg]} \abs{\expect{\mathds{1}_{X_i} - \mathds{1}_{Y_i}}}\), with \(\expect{\mathds{1}_{X_i} - \mathds{1}_{Y_i}} = x_i - y_i\).
    If $\sum_{i \in [\maxDeg]}\abs{x_i - y_i} \ge S_{\maxDeg,b}/(1000 \sqrt{b})$, then the claim directly follows. 
    Suppose now that $\sum_{i \in [\maxDeg]}\abs{x_i - y_i} < S_{\maxDeg,b}/(1000 \sqrt{b})$.
    By the triangular inequality, it holds that \(\sum_{i \in [\maxDeg]}\abs{\mathds{1}_{X_i} - \mathds{1}_{Y_i}} \ge \abs{\sum_{i \in [\maxDeg]}\mathds{1}_{X_i} - \mathds{1}_{Y_i}}\), and observe that \(\abs{\sum_{i \in [\maxDeg]}\mathds{1}_{X_i} - \mathds{1}_{Y_i}} = \abs{b - Y}\).
    Hence, $\expectView{T-1}{E^{(0)}_{\M\U}(v)} \ge \expect{\abs{b - Y}}$.
    By applying \Cref{lem:probability-deviation}, we obtain that $\expect{\abs{b - Y}} \ge S_{\maxDeg,b}(v)/(1000\sqrt{b})$. Hence, the claim follows.
\end{proof}

\paragraph{Algorithm $\algo_1$ matches many edges.}
We now prove that $\algo_1$ has low badness, implying \Cref{lem:onestep}. On a high level, we consider the edges $\mathcal{E}$ in which $\algo_0$ outputs $\M \M$ but $\algo_1$ does not. All these edges must be \emph{non-preferred} by at least one endpoint. Moreover, for each node $v$, $S_{\maxDeg,b}(v)$ is the expected number of $\M$ that $\algo_0$ outputs on the non-preferred directions of $v$, when we consider all possible extensions of the radius-$(T-1)$ view of $v$ into radius-$T$ views. Hence, we can upper bound the expected size of $\mathcal{E}$ as a function of the sum of the $S_{\maxDeg,b}(v)$ values, which is given by \Cref{lemma:analysis:algorithm-A0-expected-MU}. If $\mathcal{E}$ is small, we get that many of the edges that are matched by $\algo_0$ are also matched by $\algo_1$.

\begin{lemma}\label{lemma:analysis:algorithm-A1-expected-MM}
    The expected badness of $\algo_1$ is upper bounded by $p_0(1 + 4000\sqrt{b})$.
\end{lemma}
\begin{proof}
    By the definition of badness, we need to prove that  \(\expect{E^{(1)}_{\M\M}} \ge b\frac{n}{2}(1 - p_0(1 + 4000\sqrt{b}))\).
    We define \(E^{(0,1)}_{\M\M}\) to be the number of edges that are labeled \(\M\M\) by both \(\algo_0\) and \(\algo_1\).
    It holds that
    \[
        \expect{E^{(1)}_{\M\M}} \ge \expect{E^{(0,1)}_{\M\M}}.
    \]
    Now, define \(E^{(\text{wrong})}_{\M\M}\) to be the number of edges that \(\algo_0\) labels \(\M\M\) but \(\algo_1\) does not label \(\M\M\).
    At the same time
    \begin{align*}
        \expect{E^{(0,1)}_{\M\M}} & = \expect{E^{(0)}_{\M\M}} - \expect{E^{(\text{wrong})}_{\M\M}}.
    \end{align*}
    Let \(H^{(\text{wrong})}_{\M}\) be the number of half-edges that \(\algo_0\) labels \(\M\) and \(\algo_1\) labels \(\U\).
    It follows that
    \begin{align*}
        \expect{E^{(\text{wrong})}_{\M\M}} & \le \expect{H^{(\text{wrong})}_{\M}}.
    \end{align*}

    Let $v$ be a node where $\algo_1$ outputs $\U$ on the $i$th port for some $i$.
    Conditioned on the $(T-1)$-view of $v$, algorithm $\algo_0$ outputs $\M$ on the $i$th port of $v$ with probability $x_i(v)$.
    Recall that \(X_i(v)\) is the indicator random variable that outputs \(1\) if \(\algo_0\) outputs \(\M\) on the \(i\)th port of \(v\), and \(0\) otherwise.
    By definition, the sum over all nodes \(v\) of all but the \(b\) largest values of $x_i(v)$ is $\sum_{v \in V} S_{\maxDeg,b}(v)$.
    Let \(X_i'(v)\) be the indicator random variable that outputs \(1\) if both \(\algo_0\) and \(\algo_1\) output \(\M\) on the \(i\)th port of \(v\), and \(0\) otherwise. 
    In other words, $X'_i(v)$ represents the logical \emph{AND} between $X_i(v)$ and the indicator random variable of the event that $\mathcal{A}_1$ outputs $\M$ on the $i$-th port of $v$.
    We recall that $\mathcal{A}_1$ outputs $\M$ on the $i$-th port of $v$ if and only if \(i\) is a preferred direction for $v$.
    Note that, by linearity of expectation, \(\expect{H^{(\text{wrong})}_{\M}} =  \sum_{v \in V} \expect{\sum_{i \in [\maxDeg]} X_i(v) - X_i'(v)}\).
    By the law of total expectation, we get that 
    \(\expect{\sum_{i \in [\maxDeg]} X_i(v) - X_i'(v)} = \expect{\expectView{T-1}{\sum_{i \in [\maxDeg]} X_i(v) - X_i'(v)}}\), where the randomness of \(\expectView{T-1}{}\) is given by the private random bits of the input graph \(G\) in the radius-\((T-1)\) view of \(v\) plus shared randomness.
    Conditioning on the radius-\((T-1)\) view of \(v\), \(\expectView{T-1}{X_i(v) - X_i'(v)} = 0\) if \(i\) is a preferred direction for \(v\), and \(\expectView{T-1}{X_i(v) - X_i'(v)} = x_i(v)\) otherwise.
    Hence, \(\expectView{T-1}{\sum_{i \in [\maxDeg]} X_i(v) - X_i'(v)} = S_{\maxDeg,b}(v)\) and, by linearity of expectation, \(\sum_{v \in V} \expect{\expectView{T-1}{\sum_{i \in [\maxDeg]}  X_i(v) - X_i'(v)}} = \expect{\sum_{v \in V} S_{\maxDeg,b}(v)}\).
    Therefore, \(\expect{H^{(\text{wrong})}_{\M}} = \expect{\sum_{v \in V} S_{\maxDeg,b}(v)}\).
    Moreover, by linearity of expectation, \cref{lemma:analysis:algorithm-A0-expected-MU}, and the law of total expectation, it holds that $\expect{\sum_{v \in V} S_{\maxDeg,b}(v)} \le 2000\sqrt{b} \cdot \expect{E^{(0)}_{\M\U}}$. We thus get $\expect{H^{(\text{wrong})}_{\M}}\le 2000\sqrt{b} \cdot \expect{E^{(0)}_{\M\U}}$.

    By assumption on \(\algo_0\), the badness of $\algo_0$ is at most $p_0$, and hence  \(\expect{E^{(0)}_{\M\M}} \ge b\frac{n}{2}(1-p_0)\).
    It follows that 
    \begin{align*}
        \expect{E^{(0)}_{\M\U}} & \le bn - 2\expect{E^{(0)}_{\M\M}} \\
        &\le bnp_0.
    \end{align*}
    Therefore, 
    \begin{align*}
        \expect{E^{(1)}_{\M\M}} & \ge \expect{E^{(0)}_{\M\M}} - \expect{E^{(\text{wrong})}_{\M\M}} \\
        & \ge \expect{E^{(0)}_{\M\M}} - \expect{H^{(\text{wrong})}_{\M}} \\
        & \ge b\frac{n}{2}(1 - p_0) - 2000b^{\frac{3}{2}}np_0 \\
        & = b\frac{n}{2}(1 - p_0 - 4000\sqrt{b}p_0)\\
        &= b\frac{n}{2}(1 - (p_0 + 4000\sqrt{b}p_0)).
    \end{align*}
\end{proof}

\subsection{Badness of any zero-round algorithm}\label{ssec:zero}
We now lower bound the badness of any zero-round algorithm for $b$-grabbing.
\begin{lemma}\label{lem:zero-bound}
    Let $b \le \maxDeg/2$.
    Any $0$-round algorithm for $b$-grabbing has expected badness lower bounded by $1/2$.
\end{lemma}
\begin{proof}
    Since the initial knowledge of all nodes is the same (recall that we consider a setting without IDs, since IDs can be randomly generated), any $0$-round algorithm can be seen as a function that maps the given randomness into a subset of size exactly $b$ of $\{1,\ldots,\maxDeg\}$. Since each node grabs its incident edges without coordination, an edge is matched with probability $(b/\maxDeg)^2$, and the expected number of matched edges (i.e., the expected quality of the solution) is at most $\maxDeg \frac{n}{2} (b/\maxDeg)^2 = b \frac{n}{2}(1 - \frac{\maxDeg-b}{\maxDeg})$. Thus, the expected badness is lower bounded by $\frac{\maxDeg-b}{\maxDeg} \ge 1/2$.
\end{proof}

\subsection{Recursion}\label{ssec:iteration}
By iteratively applying \Cref{lem:onestep} for $T$ times, we obtain the following result.
\begin{lemma}\label{lem:recursion}
    Let $T \le \frac{\varepsilon}{4} \log_\maxDeg n$.
    Let $\algo$ be a no-error $T$-round algorithm that solves $b$-grabbing with expected badness upper bounded by $p$. Then, there exists a no-error $0$-round algorithm $\algo'$ that solves  $b$-grabbing with expected badness upper bounded by $(c \cdot \sqrt{b})^T \cdot p$, for some universal constant $c > 0$.
\end{lemma}
By combining \Cref{lem:zero-bound} with \Cref{lem:recursion}, we obtain the following.
\begin{lemma}\label{lem:bound-grabbing}
    Let $b \le \maxDeg/2$. Any no-error algorithm $\algo$ that solves $b$-grabbing  with expected badness upper bounded by $p$ requires at least $T = \min\{\frac{\varepsilon}{4} \log_\maxDeg n, \frac{\log (1/(2p))}{\log(c \sqrt{b}) }\}$.
\end{lemma}
\begin{proof}
    Proceeding by disjoint cases, we have two scenarios to consider. If $T > \frac{\varepsilon}{4} \log_\maxDeg n$, then, trivially, the claim follows. 
    Suppose now that $T \leq \frac{\varepsilon}{4} \log_\maxDeg n$.
    By \Cref{lem:recursion}, given a $T$-round algorithm $\algo$ that solves $b$-grabbing  with expected badness upper bounded by $p$, we can construct a $0$-round algorithm $\algo'$ with expected badness upper bounded by $(c \cdot \sqrt{b})^T \cdot p$. By \Cref{lem:zero-bound}, we must have $(c \cdot \sqrt{b})^T \cdot p \ge 1/2$.
    This implies that $T \geq \frac{\log (1/(2p))}{\log(c \sqrt{b}) }$.
    Hence, the claim follows for this case, too.
\end{proof}

We now prove a lower bound for maximal $b$-matching on $\mathcal{G}$.
\begin{lemma}\label{lem:lb-family}
    Any randomized \local algorithm $\algo$ for maximal $b$-matching with failure probability at most $b \frac{\log \maxDeg}{\maxDeg}$ on $\mathcal{G}$ requires $\Omega(\min\{\log_{1+b} \maxDeg, \log_\maxDeg n\})$, for any \(b \in o(\maxDeg / \log \maxDeg)\).
\end{lemma}
\begin{proof}
    By \Cref{cor:starting}, if there exists a $T$-round algorithm for maximal $b$-matching, for $T \le \frac{\varepsilon}{4} \log_\maxDeg n$, then there exists a no-error $T$-round algorithm $\algo'$ that solves $b$-grabbing with badness at most $p = \Theta(b \frac{\log \maxDeg}{\maxDeg})$ in any $G \in \mathcal{G}$.
    By \Cref{lem:bound-grabbing}, we obtain that $T \ge \min\{\frac{\varepsilon}{4} \log_\maxDeg n, \frac{\log (1/(2p))}{\log(c \sqrt{b}) }\}$.
    Hence, the claim follows.
\end{proof}

\subsection{Obtaining a lower bound for trees}\label{ssec:trees}
A \emph{locally checkable problem} is a problem for which the validity of a solution is defined via local constraints, or in other words, there exists a constant-time algorithm called \emph{verifier} satisfying the following:
\begin{itemize}
    \item If a given solution is correct, the verifier \emph{accepts} at every node.
    \item If a given solution is not correct, the verifier \emph{rejects} on at least one node.
\end{itemize}
Clearly, maximal $b$-matching is a locally checkable problem.
It is folklore that, for locally checkable problems, any lower bound for high-girth graphs implies a lower bound for trees as well (see, e.g., \cite{linial-1992-locality-in-distributed-graph-algorithms,balliu-brandt-etal-2020-classification-of-distributed,chang-kopelowitz-pettie-2016-an-exponential-separation,balliu-ghaffari-etal-2022-node-and-edge-averaged}). However, no formal proof for this generic statement has been given in previous works. Hence, for completeness, we now prove the following statement. 
Note that we will use $p$ to denote the failure probability of an algorithm and no longer for the badness.

\begin{lemma}\label{lem:lift-trees}
    Let $T(\maxDeg,n) = \min(f(\maxDeg), \frac{\varepsilon}{4} \log_\maxDeg n)$ for some function $f$ and some constant $\varepsilon > 0$.
    Let $\mathcal{G}$ be a family of $\maxDeg$-regular graphs, such that each $n$-node graph $G \in \mathcal{G}$ has girth at least $\varepsilon\log_\maxDeg n$.
    
    Let $\Pi$ be a locally checkable problem that, for any algorithm $\algo$ with failure probability at most $p$, for all $n$ large enough, there exists a graph $G \in \mathcal{G}$ of $n$ nodes where $\algo$ requires at least $T(\maxDeg,n)$, where $T(\maxDeg,n) < \frac{\varepsilon}{4} \log_\maxDeg n$.

    Then, $\Pi$, on $\maxDeg$-regular trees, for any $n$ large enough requires at least $\gamma \cdot T(\maxDeg,n)$ rounds, for any algorithm with failure probability at most $p' = p/n^\gamma$ for any constant $\gamma > 0$.
\end{lemma}
\begin{proof}
    We first prove that $\Pi$, on $\maxDeg$-regular trees, for all $n$ large enough, requires $T(\maxDeg,n)$ rounds, for any algorithm with failure probability at most $p' = p/n$. 
    By contradiction, suppose that an algorithm $\algo$ for trees contradicting the stated lower bound exists.
    Hence, there are infinitely many values of $n$ and $\maxDeg$-regular trees of $n$ nodes, where $\algo$ terminates in strictly less than $T(\maxDeg,n)$ rounds and fails with probability at most $p'$. Let $N$ be the set of such values of $n$.
    
    By assumption, there must exist a graph $G$ of $n$ nodes in $\mathcal{G}$ such that $n \in N$ and, if we run $\algo$ on $G$, it must fail with probability strictly larger than $p$. This implies that there exists at least one node $v \in V(G)$ at which the verifier of $\Pi$ rejects with probability strictly larger than $p' = p / n$. Take the radius-$T(\maxDeg,n)$ neighborhood $\mathcal{N}$ of $v$. By the girth assumption, $\mathcal{N}$ is a tree of at most $n$ nodes, and hence it can be extended into a tree $\mathcal{T}$ of exactly $n$ nodes such that the radius-$T(\maxDeg,n)$ neighborhood of $v$ in $\mathcal{T}$ and in $G$ are the same. Hence, $\algo$ must fail on $\mathcal{T}$ with probability strictly larger than $p' = p / n$, a contradiction.

    We now use a standard \emph{lie-about-n} argument to replace $n$ with $n^\gamma$.
    We prove that $\Pi$, on $\maxDeg$-regular trees, for all $n$ large enough, requires at least $\gamma \cdot T(\maxDeg,n)$ rounds, for any algorithm with failure probability at most $p/n^\gamma$ for any $\gamma > 0$. By contradiction, suppose that a value $\gamma < 1$ and an algorithm $\algo$ for trees contradicting the stated lower bound exists.
    Recall that $\algo$ takes as input the value of $n$. If we run $\algo$ on a graph of $n$ nodes by giving as input $n' = n^{1/\gamma}$, we would obtain the following:
    \begin{itemize}
        \item The runtime of the algorithm becomes strictly less than $\gamma \cdot T(\maxDeg, n^{1/\gamma}) \le \gamma \frac{\varepsilon}{4} \log_\maxDeg (n^{1/\gamma}) = \frac{\varepsilon}{4} \log_\maxDeg n$. Hence, since any node running $\algo$ does not see the whole graph, the algorithm cannot detect that the given $n$ is not the correct one, and hence it must work correctly.
        \item The runtime of the algorithm is strictly less than $\gamma \cdot T(\maxDeg, n^{1/\gamma}) = \gamma \cdot \min(f(\maxDeg), \frac{\varepsilon}{4} \log_\maxDeg n^{1/\gamma}) \le T(\maxDeg,n)$.
        \item The failure probability becomes $p / (n')^{\gamma} = p/n$.
    \end{itemize}
    Hence, we reach a contradiction with the case in which $\gamma=1$.
\end{proof}
By combining \Cref{lem:lb-family} and \Cref{lem:lift-trees}, and by taking $\maxDeg = 2^{\sqrt{\log n \log (1+b)}}$, we obtain the following theorem.
\begin{theorem}\label{theorem:analysis:lower-bound-relaxed-b-matching}  
    Any randomized \local algorithm $\algo$ that solves maximal $b$-matching on \(\maxDeg\)-regular trees with failure probability at most $b \frac{\log \maxDeg}{\maxDeg} / n^\gamma$ requires $\Omega(\min\{\log_{1+b} \maxDeg, \log_\maxDeg n\})$ and $\Omega(\sqrt{\log_{1+b} n})$ rounds, for any constant $\gamma > 0$ and for any \(b \in o(\maxDeg / \log \maxDeg)\).
\end{theorem}
Observe that the standard requirement for the failure probability of a randomized algorithm is to be upper bounded by $1/n$. \Cref{theorem:analysis:lower-bound-relaxed-b-matching} implies a lower bound for such algorithms by picking a suitable value for $\gamma$.

\section{\texorpdfstring{Lower bound for \((\maxDeg + k)\)-edge coloring}{Lower bound for (Delta+k)-edge coloring}}\label{sec:edge-coloring-lb}
In this section, we show a lower bound for the \((\maxDeg+k)\)-edge coloring problem.
The problem requires assigning a color to each edge, from a palette of $\maxDeg + k$ colors, such that edges incident to the same node are assigned pairwise distinct colors.

\begin{theorem}\label{theorem:analysis:edge-coloring-lower-bound}
    Let $\delta \in (0,1]$ be any constant.
For any $k \le \maxDeg^{1-\delta}$, any randomized \local algorithm $\algo$ that solves $(\maxDeg+k)$-edge coloring on trees with failure probability at most $\frac{1}{\maxDeg^\delta} / n^\gamma$ requires $\Omega(\min\{\log \maxDeg, \log_\maxDeg n\})$ and $\Omega(\sqrt{\log n})$ rounds, for any constant $\gamma > 0$.
\end{theorem}
\begin{proof}
    Let $T$ be the runtime of an algorithm solving $(\maxDeg+k)$-edge coloring.
    We prove our lower bound by showing that a $(\maxDeg+k)$-edge coloring can be converted into a solution for $1$-grabbing with small badness in $0$ rounds.
    Assume we are given a solution for $(\maxDeg+k)$-edge coloring. Sample a color class uniformly at random from $\{1,\ldots,\maxDeg+k\}$ using shared randomness. The expected number of edges of this color is at least $\frac{\maxDeg n/2}{\maxDeg+k} = \frac{n}{2}(1 - \frac{k}{\maxDeg + k})$. Hence, we can solve $1$-grabbing with expected badness upper bounded by $\frac{k}{\maxDeg+k} \le 1 / \maxDeg^{\delta}$ in $T$ rounds. 
    By \Cref{lem:bound-grabbing}, we obtain that $T \ge \min\{\frac{\varepsilon}{4} \log_\maxDeg n, \frac{\log (1/(2p))}{\log(c) }\}$, where $p = 1 / \maxDeg^{\delta}$, and $c>1$ is the constant given by \Cref{lem:onestep}, hence obtaining the claimed lower bound as a function of $n$ and $\maxDeg$ for the family $\mathcal{G}$ of graphs for algorithms with failure probability at most $p$.
    By applying \Cref{lem:lift-trees}, the stated lower bound holds for trees. The lower bound as a function of $n$ is obtained by taking $\maxDeg = 2^{\sqrt{\log n}}$.
\end{proof}

%% file: trunk/acks.tex
\section*{Acknowledgments}

We would like to thank Maxime Flin for spotting a wrong constant in a proof.

\noindent
Alkida Balliu was supported by FIS-2024-03606 - A Robust Theory of Distributed Computation CUP: D53C25002470001 - Decreto Direttoriale del 21 Novembre 2024, N. 1802 - Procedura competitiva per lo sviluppo delle attività di ricerca fondamentale, a valere sul fondo italiano per la scienza 2024-2025 - (bando FIS 3) - Decreto ammissione al finanziamento prot. MUR 0018010 del 12-11-2025.

\noindent
Francesco d'Amore was supported by Decreto MUR n. 47/2025, CUP D13C25000750001.
This work was partially supported by the MUR (Italy) Department of Excellence 2023--2027 for GSSI.

%% file: trunk/appendix.tex
\appendix

\section{Probabilistic inequalities}\label{app:probabilistic-inequalities}

\begin{lemma}[Simplified Khintchine inequality \cite{haagerup1981}]\label{lemma:app:khintchine}
    Let \(\varepsilon_1, \ldots, \varepsilon_n \sim \rademacher(1/2)\) be i.i.d.\ random variables. 
    Let \(x_1, \ldots, x_n \in \reals\).
    It holds that
    \begin{align*}
        \frac{1}{\sqrt{2}}\left(\sum_{i=1}^n \abs{x_i}^2\right)^{\frac{1}{2}} \le \expect{\abs{\sum_{i=1}^n x_i\varepsilon_i}}
        \le \left(\sum_{i=1}^n \abs{x_i}^2\right)^{\frac{1}{2}}.
    \end{align*} 
\end{lemma}

\begin{lemma}[Paley-Zygmund inequality]\label{lemma:app:paley-zygmund}
    Let \(Z \ge 0\) be any random variable with finite variance.
    For any \(0 \le \lambda \le 1\), it holds that 
    \begin{align*}
        \Pr[Z \ge \lambda \expect{Z}] \ge (1 - \lambda)^2 \frac{\expect{Z}^2}{\expect{Z^2}}.
    \end{align*}
\end{lemma}

\section{Proof of \texorpdfstring{\Cref{lem:probability-deviation}}{Probability Lemma}}\label{app:probability-deviation}
    Let \(S' = \sum_{i \in [\maxDeg]} y_i\). Observe that $S' = \expect{Y}$.
    Let $Z_i = Y_i - y_i$, and let $Z = \sum_{i \in [\maxDeg]} Z_i$.
    Note that $\expect{Z} = 0$, and that $\expect{\abs{Z}}$ counts how much $Y$ deviates from \(S'\) in expectation.
    Hence,
    \[
    \expect{|b - Y|} = \expect{|(Y - S') - (b - S') } = \expect{|Z - (b - S')|} \ge \expect{|Z| - |b-S'|},
    \]
    where the last inequality holds by the triangle inequality.
    Since $|b- S'| \le \sum_{i \in [\maxDeg]}\abs{x_i - y_i} < S_{\Delta,b}/(1000 \sqrt{b})$, in order to prove our claim, it is sufficient to prove that $\expect{|Z|} \ge S_{\Delta,b}/(100 \sqrt{b})$. 

    \paragraph{Taking two copies of $Z$.}
    To prove such an inequality, we use the following approach.
    We define the events $Z_i'$ so that they have the same probability as $Z_i$, but are independent of $Z_i$, and define \(Z' = \sum_{i \in [\maxDeg]} Z_i'\). That is, $Z'$ is an independent copy of $Z$.

    By the triangular inequality, it holds that 
    \[
        \expect{|Z - Z'|} \le \expect{|Z|} + \expect{|Z'|} = 2 \expect{|Z|}.
    \]
    Hence, $\expect{|Z|} \ge \expect{|Z -  Z'|}/2$. Observe that $W_i = Z_i - Z_i'$ is $1$ with probability $y_i(1 - y_i)$, $-1$ with probability $y_i(1- y_i)$, and $0$ otherwise.
    Thus, $W_i$ can be written as the product of the independent random variables $\varepsilon_i$ and $J_i$, where $\varepsilon_i \sim \rademacher(1/2)$ and
    $J_i \sim \bern(2y_i(1 - y_i))$ (where $\rademacher$ and $\bern$ are the Rademacher and the Bernoulli distributions, respectively).     Let $J = \sum_{i \in [\maxDeg]} J_i$.

    Observe that, for a set $[m]$ of integers, by Khintchine's inequality (\cref{lemma:app:khintchine}) it holds that 
    \(\expect{\abs{\sum_{i \in [m]} \varepsilon_i}} \ge \sqrt{m/2}\).
    Since for all $i$, $J_i$ is either $0$ or $1$, it follows that 
    \(\expect{\abs{\sum_{i \in [\maxDeg]} \varepsilon_i J_i} \st J} \ge \sqrt{\frac{J}{2}}\).

    Hence,
    \[
        \expect{|Z-Z'|} = \expect{\abs{\sum_{i \in [\maxDeg]} W_i}} = 
        \expect{\abs{\sum_{i \in [\maxDeg]} \varepsilon_i J_i }} = \expect{\expect{\abs{\sum_{i \in [\maxDeg]} \varepsilon_i J_i } \st J}} \ge \expect{\sqrt{\frac{J}{2}}},
    \]
    where the last equality follows by the law of total expectation.
    Hence, $\expect{|Z|} \ge \expect{\sqrt{J}} / (2\sqrt{2})$, and in the following we provide a lower bound for $\expect{\sqrt{J}}$.

    Let $\Lambda =  \sum_{i = 1}^{\maxDeg} 2y_i(1-y_i)$. Note that $\expect{J} = \Lambda$. 
    We have two cases: \(\Lambda \le 1\) or \(\Lambda > 1\).

    \paragraph{\boldmath The case $\Lambda \le 1$.}

    Suppose \(\Lambda \le 1\).
    Then, 
    \[
    \expect{\sqrt{J}} \ge \pr{J \ge 1} 
         = 1 - \prod_{i = 1}^\maxDeg (1 - 2y_i(1-y_i)) 
         \ge 1 - e^{-\sum_{i = 1}^\maxDeg 2y_i(1-y_i)} 
         = 1 - e^{-\Lambda} 
         \ge \frac{\Lambda}{2},
    \]
    where in the last inequality we used that \(1 - e^{-x} \ge x/2\) for any \(x \in [0,1]\) and the fact that $\Lambda \le 1$.

    \paragraph{\boldmath The case $\Lambda > 1$.}
    If $\Lambda > 1$, first observe that 
    \[
        \expect{J^2} - \expect{J}^2  = \variance{J} 
         = \sum_{i = 1}^\maxDeg \variance{J_i} 
         = \sum_{i = 1}^\maxDeg 2y_i(1-y_i)(1 - 2y_i(1-y_i)) 
         \le \sum_{i = 1}^\maxDeg 2y_i(1-y_i) 
         = \expect{J} 
         = \Lambda,
    \]
    where the second equality holds by independence. Thus, \(\expect{J^2} \le \Lambda + \Lambda^2\).
    Hence, by the Paley-Zygmund's inequality (\Cref{lemma:app:paley-zygmund}) we obtain that 
    \[
    \pr{J \ge \Lambda / 2}  \ge \frac{\Lambda}{4(1+\Lambda)} 
         \ge \frac{1}{8}.
    \]
    Finally, it holds that
    \[
        \expect{\sqrt{J}} \ge \sqrt{\frac{\Lambda}{2}} \pr{J \ge \frac{\Lambda}{2}} \ge \frac{1}{8}\sqrt{\frac{\Lambda}{2}}. 
    \]

    \paragraph{Combining the two cases.}
    By combining the two cases, we obtain 
    \[
        \expect{|Z|} \ge \expect{\sqrt{J}} / (2\sqrt{2}) \ge \frac{1}{32}\min\{\Lambda,\sqrt{\Lambda}\}. 
    \]
    Since $\Lambda = \sum_{i = 1}^\maxDeg 2y_i(1-y_i) \le \sum_{i = 1}^\maxDeg \min\{y_i, 1-y_i\}$, we get that \[
        \expect{|Z|} \ge \frac{1}{32}\min \left\{\sum_{i = 1}^\maxDeg \min\{y_i, 1 - y_i\},\sqrt{\sum_{i = 1}^\maxDeg \min\{y_i, 1 - y_i\}}\right\}.
    \]

    \paragraph{\boldmath Bound on $\sum_{i \in [\maxDeg]} \min\{x_i,1-x_i\}$.}
    Before proving a bound on $\sum_{i \in [\maxDeg]} \min\{y_i,1-y_i\}$, we first prove the following:
    \[
    \sum_{i \in [\maxDeg]} \min\{x_i,1-x_i\} \ge S_{\Delta,b} \overset{\text{def}}{=} \sum_{b+1 \le i \le \Delta} x_i.
    \]
    W.l.o.g., assume $x_1 \ge x_2 \ge ... \ge x_\Delta$.
    Let \(i^\star\) be the first index such that \(x_i < 1 - x_i\), and $i^\star = \Delta+1$ if no such index exists. Observe that for all \(i < i^\star\) it holds that \(x_i \ge 1/2\).   
    Since $\sum x_i = b$, it follows that \(i^\star \le 2b + 1\).
     If \(i^\star \le b+1\), then the inequality trivially holds, as \(x_i < 1 - x_i\) for all \(i > b\).
    Otherwise, if $i^\star > b + 1$,
    \begin{align*}
       \sum_{i = 1}^\maxDeg \min\{x_i, 1 - x_i\} & = \sum_{i = 1}^{b} (1-x_i) + \sum_{i = b+1}^{i^\star - 1} (1-x_i) + \sum_{i = i^\star}^{\maxDeg} x_i. 
    \end{align*}    

    Hence,
    \begin{align}
        \sum_{i = 1}^\maxDeg \min\{x_i, 1 - x_i\} - S_{\maxDeg,b} & = \sum_{i = 1}^{b} (1 - x_i) + \sum_{i = b+1}^{i^\star - 1} (1 - 2x_i) \nonumber \\
        & = b - S_b + (i^\star - (b+1)) - 2\sum_{i = b+1}^{i^\star - 1} x_i \nonumber \\
        & = i^\star - 1 - S_b - 2 (S_{i^\star - 1} - S_b) \nonumber \\
        & = i^\star - 1 + S_b - 2S_{i^\star - 1} \nonumber \\
        & \ge i^\star - 1 + \frac{b^2 - 2b(i^\star - 1)}{i^\star -1} \label{eq:analysis:minsum}\\
        & = \frac{(i^\star - 1)^2 - 2(i^\star - 1)b + b^2}{i^\star - 1} \nonumber \\
        & = \frac{(i^\star - (b+1))^2}{i^\star - 1} \nonumber \\
        & \ge 0, \nonumber 
    \end{align}
    where \cref{eq:analysis:minsum} holds because \(S_b -2S_{i^\star - 1}\) is minimum when \(S_{i^\star - 1}\) is maximum and all \(x_{j}\) that occur in the sum \(S_{i^\star - 1}\) are equal, that is, \(S_{i^\star - 1} = b\) and \(x_{j}= b/(i^\star - 1)\).
    \paragraph{\boldmath Bound on $\sum_{i \in [\maxDeg]} \min\{y_i,1-y_i\}$.}
    By the assumption on the deviation between the $x_i$ and the $y_i$ values, we obtain that
    \[
        \sum_{i \in [\maxDeg]} \min\{y_i,1-y_i\} \ge \sum_{i \in [\maxDeg]} \min\{x_i,1-x_i\} - \frac{S_{\Delta,b}}{1000 \sqrt{b}} \ge S_{\Delta,b}  - \frac{S_{\Delta,b}}{1000 \sqrt{b}} \ge \frac{999}{1000} S_{\Delta,b}.
    \]

    \paragraph{\boldmath Bound on $\expect{|Z|}$.}
    We are now ready to bound $\expect{|Z|}$. We have that
    \[
        \expect{\abs{Z}}  \ge \frac{1}{32} \cdot \frac{999}{1000}  \min \left\{S_{\maxDeg,b}, \sqrt{S_{\maxDeg,b}}\right\} 
         \ge \frac{1}{64}  \min \left\{\frac{S_{\maxDeg,b}}{\sqrt{b}}, \sqrt{S_{\maxDeg,b}}\right\} 
         \ge \frac{S_{\maxDeg,b}}{64\sqrt{b}}
         \ge \frac{S_{\maxDeg,b}}{100\sqrt{b}}
         ,
    \]
    as required.

\section{Simpler proof for maximal matching}\label{app:simpler-proof-maximal-matching}
In this section, we provide a simple proof of \Cref{lem:probability-deviation} for the case in which $b = 1$.
Let $X_i$ be the event corresponding to $v$ outputting $\M$ towards its $i$th neighbor, let $Y_i$ be the event that the $i$th neighbor of $v$ outputs $\M$ towards $v$, and define \(x_i = \pr{X_i}\) and \(y_i = \pr{Y_i}\).
On a high level, if $x_i$ and $y_i$ differ too much, we already have high probability of obtaining $\M \U$ edges.
Hence, for simplicity, assume that, for all $i$, $x_i = y_i$. 
If the neighbors of $v$ output 2 $\M$s towards $v$, we are sure that we get an $\M \U$ edge, because $v$ outputs only one $\M$. Hence, it is sufficient to lower bound the probability of getting two $\M$s towards $v$.
W.l.o.g., assume $x_1 \ge \ldots \ge x_\Delta$.
Either $x_1$ is large, or it is small.
If $x_1$ is large, the probability of getting two $\M$s is at least $x_1 \cdot S_{\Delta,1}$, and we are done. If $x_1$ is small, then the $x_i$ values are well-distributed, and we have constant probability of getting at least two $\M$s.

\begin{lemma}
    Let \(x_1, \ldots, x_\Delta\) be non-negative real numbers. 
    Assume that $\sum{x_i} = 1$.
    Define $Y_1,\ldots,Y_\Delta$ to be independent events, where $Y_i$ has probability $y_i$, and let $M = 1 - \max\{x_1,\ldots,x_\Delta\}$.
    Let $Y$ be the number of events that occur among $Y_1,\ldots,Y_\Delta$.
    Assume that $\sum_{i \in [\maxDeg]}\abs{x_i - y_i} < M/1000$.
    Then, $\expect{|Y - 1|} \ge M / 1000$.
\end{lemma}
\begin{proof}
    Observe that 
    \[\expect{|Y - 1|} = \sum_{0\le i \le \Delta} |i - 1 |\Pr[Y = i] \ge \sum_{2 \le i \le \Delta} (i-1) \Pr[Y = i] \ge \sum_{2 \le i \le \Delta} \Pr[Y = i] = \Pr[Y \ge 2].\]
    Hence, in order to prove the claim, we show that $\Pr[Y \ge 2] \ge M/1000$.
    First observe that, for any subset $S \in [\Delta]$, it holds that 
    \[
    \Pr[ \sum_{i\in S} Y_i  \ge 1] = 1 - \prod_{i\in S} (1 - y_{i}) \ge 
    1 - \exp\left[-\sum_{i\in S} y_{i} \right] \\
        \ge \frac{1}{2}\sum_{i\in S} y_{i},
    \]
    where the first equality holds by the independence assumption, and the last inequality holds by the fact that $1 - \exp(-x) \ge x/2$ for $x \in [0,3/2]$.

    W.l.o.g., assume $x_1 \ge x_2 \ge \ldots \ge x_\Delta$.
    We consider two cases, either $x_1 > 1/10$ or $x_1 \le 1/10$.
    In the former case, 
    \[
    \Pr[Y_1 =1] > \frac{1}{10} - \frac{1}{1000},  \text{ and } 
    \]
    \[\Pr[\sum_{2\le i \le \Delta} Y_i \ge 1] \ge \frac{1}{2}\sum_{2\le i \le \Delta} y_i \ge  \frac{1}{2}\left(\sum_{2\le i \le \Delta} x_i - M /1000\right) = \frac{M}{2} - \frac{M}{2000}.
    \]
    Thus, by the independence assumption, we get that 
    \[
    \Pr[Y \ge 2] \ge \left(\frac{1}{10} - \frac{1}{1000}\right) \cdot \left(\frac{M}{2} - \frac{M}{2000}\right) \ge M/1000, \text{ as required.}
    \]

    In the latter case, observe that, for all $i \in [\Delta]$, $x_i \le 1/10$. Hence, there exists an index $k$ such that $x_1 + \ldots + x_k \ge 1/10$ and $x_{k+1} + \ldots + x_\Delta \ge 1/10$.
    By assumption on the deviation between the $x_i$ values and the $y_i$ values, we get that $y_1 + \ldots + y_k \ge 1/10 - 1/1000$ and $y_{k+1} + \ldots + y_\Delta \ge 1/10 - 1/1000$.
    By the independence assumption, we thus get that\[
        \Pr[Y \ge 2] \ge \Pr[\sum_{1\le i \le k} Y_i \ge 1] \cdot \Pr[\sum_{k+1\le i \le \Delta} Y_i \ge 1] \ge \left(\frac{1}{2}\left(\frac{1}{10} - \frac{1}{1000}\right)\right)^2 \ge \frac{1}{1000} \ge \frac{M}{1000},
    \]
    as required.
\end{proof}